\documentclass{article}
\usepackage{amsmath,amssymb,mathrsfs,amsthm}

\theoremstyle{plain}
\newtheorem{thm}{Theorem} 
 
\newtheorem{prop}{Proposition}[section]
\newtheorem{lemma}[prop]{Lemma}

\theoremstyle{definition} 
 
\newtheorem{asm}{Assumption}

\theoremstyle{remark}
\newtheorem{remark}[prop]{Remark}

\newcommand{\RR}{\mathbb{R}}
\newcommand{\CC}{\mathbb{C}}
\newcommand{\ZZ}{\mathbb{Z}}

\newcommand{\Cinf}{C^{\infty}}
\newcommand{\ve}{\varepsilon}
\newcommand{\supp}{\mathop{\mathrm{supp}}}
\newcommand{\RE}{\mathrm{Re}\,}
\newcommand{\IM}{\mathrm{Im}\,}

\begin{document}

\title{Semiclassical shape resonances for magnetic Stark Hamiltonians}
\author{Kentaro Kameoka and Naoya Yoshida}

\date{}

\maketitle

\begin{abstract}
We study shape resonances of two-dimensional magnetic Stark Hamiltonians in the semiclassical limit. 
The magnetic field is assumed to be constant and 
the scalar potential is a perturbation of a linear potential. 
Under the assumption that the scalar potential has potential wells, 
the existence of a one-to-one correspondence between shape resonances of the Hamiltonian and discrete eigenvalues of 
a certain reference operator is proved. 
This implies the Weyl law for the number of resonances and 
the asymptotic behavior of the real parts of resonances near the bottom of a potential well. 
Resonances are studied as complex eigenvalues of complex distorted Hamiltonians, 
which is defined by the complex translation outside a compact set. 
\end{abstract}

\section{Introduction}\label{sec-1}

In this paper, we study resonances of the following magnetic Stark Hamiltonian
\[
P=P(h)=\frac{1}{2}(hD_x+By)^2+\frac{1}{2}(hD_y)^2+x+V(x, y)\,\,\,\,\, \mathrm{on}\,\,\,L^2(\RR^2), 
\]
where $B>0$ is a fixed constant, $h>0$ is a small semiclassical parameter, 
$D=-i\partial$ and $V$ is a real-valued function. 
We discuss resonances generated by the potential wells and study their semiclassical ($h\to 0$) behavior. 

In quantum mechanics, complex resonances are complex numbers which correspond to quasi-steady states. 
The real part of a resonance is energy, and the imaginary part describes the decay rate. 
Resonances are also called scattering poles and are closely related to the scattering theory. 
Mathematical study of resonances is a fascinating research field  
(see, for example, Cycon, Froese, Kirsch and Simon \cite{CFKS} and Dyatlov and Zworski \cite{DZ}). 

Semiclassical analysis studies the asymptotic behavior of the differential operator of the form 
$p(x, hD)$, where $h>0$ is the small semiclassical parameter. 
This method is used for the deep study of quantum-classical correspondence for the Schr\"{o}dinger operators.
Microlocal analysis and pseudodifferential operators are central methods for semiclassical analysis
(see, for example, Dimassi and Sj\"{o}strand \cite{DS} and Zworski \cite{Z}).

Our operator $P(h)$ with $V=0$ describes a quantum particle in the $xy$-plane 
in the constant magnetic field in the $z$-direction and 
the constant electric field in the $x$ direction. 
The particle exhibits the drift motion which is perpendicular to the electric field. 

For comparison, we set 
\[
P(h;B,\omega)=\frac12(hD_x+By)^2+\frac12(hD_y)^2+\omega x+V(x, y)\,\,\,\,\, \mathrm{on}\,\,\,L^2(\RR^2). 
\]
In the decaying potential case ($\omega=0$), it is well-known that the essential  spectrum of the operator $P(h;B,0)$ is given by
$$
\sigma_{ess}\bigl(P(h;B,0)\bigr)=\bigcup_{n\in\ZZ_{\ge 0}}(n+\frac12)hB.
$$ 
The numbers $(n+\frac12)hB, n\in\ZZ_{\ge 0}=\{0,1,2,\cdots\}$, 
called Landau levels, are eigenvalues of infinite multiplicity (see, for instance, \cite{AHS}). 
Outside the Landau levels, eigenvalues with finite multiplicity may appear caused by the perturbation $V$.
The asymptotic distribution of eigenvalues has been studied in various regimes.  
In particular, the semiclassical asymptotics of eigenvalues is studied in 
Helffer and Sj\"{o}strand \cite{HeSjo2}, Ivrii \cite{Iv}, Raikov \cite{Rai}.   

For $\omega\not=0$, the situation completely changes and $\sigma_{ess}\bigl(P(h;B,\omega)\bigr)=\RR$. 
Note that the absence of embedded eigenvalues is studied by
Adachi and Tsujii \cite{AT}, Assel, Dimassi and Fernandez \cite{ADF},  
Dimassi, Kawamoto and Petkov \cite{DKP}, Dimassi and Petkov \cite{DP2}. 
From the physical point of view, it is expected that $V(x,y)$ creates resonances $z\in\CC, \IM z\leq 0$, 
and it is natural to study the distribution of resonances.
To the best of our knowledge, only a few works treat magnetic Stark resonances.
The shape resonances for the magnetic Stark Hamiltonian with $B$ large enough 
have been studied by Wang \cite{Wa1}.
Dimassi and Petkov \cite{DP1} established 
a connection between the resonance and the spectral shift function and 
obtained a trace formula for resonances of the magnetic Stark Hamiltonian with $B$ large enough. 
Ferrari and Kovarik showed Gaussian small estimate of the width of resonances when $\omega\to 0$
in \cite{Fe-Ko0}, 
and proved the resonance expansion of the matrix elements of the propagator for large time in \cite{Fe-Ko}. 
% showed that the width of resonances is exponentially decaying with respect to $\omega^{-1}$. 
% More precisely, they
% show that, if $z$ is a resonance then there exists two constants $C_0, C_1$ such that :
% $$
% C_1e^{-C_0\frac{B}{\omega^2}} \leq \IM z \leq 0.
% $$
We mention the two-parameter problem of $h\to 0$ and $B\to \infty$. 
Ivrii \cite{Iv} studied the distribution of eigenvalues of $P(h; B, 0)$ in this limit. 
The two-parameter problem for resonances of magnetic Stark Hamiltonians $P(h; B, \omega)$ 
is an open problem.   

Resonances of magnetic Stark Hamiltonians are defined by 
the complex translation on the whole plane in Dimassi and Petkov \cite{DP1}, Ferrari and Kovarik \cite{Fe-Ko0}, 
and by complex scaling on the whole plane in Wang \cite{Wa1}. 
In this paper, we define resonances of magnetic Stark Hamiltonians by 
the complex translation outside a compact set. 
This allows us to discuss non-globally analytic potentials. 
Recall that the resonances are described by the trapped set of the classical Hamilton flow 
in the semiclassical limit. 
Thus a complex distortion outside a compact set is more suitable for the semiclassical study of resonances 
since it does not distort the trapped set. 
In this case, rigorous proofs of the basic operator-theoretic properties of the distorted operator are 
more involved than the case of the global complex translation. 
This is one of our aims of this paper and we discuss it in Section~\ref{sec-2}.  

% We finally mention the two-parameter problem of $h\to 0$ and $B\to \infty$. 
% We set $\mu=hB$ and $\varepsilon=B^{-1}h$. 
% V. Ivrii studied the distribution of eigenvalues of $P(h; B, 0)$ in the interval $[a,b]$ when $\mu\leq1$ and $h\rightarrow0$. 
% Here $[a,b]$ is a disjoint interval of the essential spectrum of $P_{h,B}$. 
% On the other hand, if If $\mu=hB\rightarrow\infty$ and $\varepsilon\rightarrow0$, 
% the canonical transformation shows that the operator $\mu^{-1}P_{h,B}$ is  unitarily equivalent to 
% $$
% \tilde{P}(\varepsilon,\mu)=(D_x^2+x^2)\otimes I_y+\mu^{-1}V^w(\varepsilon),
% $$
% where $V^w$ is an $\varepsilon$-pseudodifferential operator.
% Thus we can apply the method of Feshbach reduction to the operator $P(h; B, 0)$ 
% to obtain the distribution of eigenvalues near each Landau levels. 
% The two-parameter problem for resonances of magnetic Stark Hamiltonians $P(h; B, \omega)$ is an open problem.   

We are interested in whether the width of resonances are exponentially 
small with respect to $h$.
To define complex resonances of $P$, we assume the following. 
\begin{asm}\label{asm-res}
There exist $R_0>0$ and $\delta_0>0$ such that the following hold. 

\noindent
(i). $V\in \Cinf(\RR^2; \RR)$ has an analytic continuation with respect to $x$ to the region 
\[
\Omega_{R_0, \delta_0}=\{(x, y)\in \CC \times \RR|\,|\RE x|^2+y^2>R_0^2, |\IM x|<\delta_0\}
\]
and belongs to $\Cinf$-class with respect to $(x, y)$ in that region. 

\noindent
\vspace{0.2cm}
(ii).  $\lim_{|(x, y)|\to \infty}\partial V(x, y)=0$ in $\Omega_{R_0, \delta_0}$. 

\noindent
\vspace{0.2cm}
(iii). $\sup_{y\in \RR} |V(0, y)|<\infty$. 

\end{asm}

Under this assumption, we can define the resonances of $P$ as follows.  
We set $R_{+}(z)=(z-P)^{-1}$ for $\IM z>0$.
Then for any non-zero $\chi_1, \chi_2\in \Cinf_c(\RR^2)$, the cutoff resolvent $\chi_1 R_{+}(z)\chi_2$ has 
a meromorphic continuation from the upper half plane to the region 
$\{z\in \CC|\, \IM z>-\delta_0\}$. Then the poles of $\chi_1 R_{+}(z)\chi_2$ are resonances of $P$. 
The set of resonances is denoted by $\mathrm{Res}(P)$.  
We define the multiplicity $m_z$ of resonance by
\[
m_z=\mathrm{rank}\frac{1}{2\pi i}\oint_{|\zeta-z|=c} \chi_1 R_+(\zeta) \chi_2 d\zeta 
\]
for $0<c\ll 1$. Then $m_z$ is independent of the choice of $\chi_1, \chi_2$. 
These are proved in Proposition~\ref{prop2-3}.

We set the classical Hamiltonian 
\[
p(x, y, \xi, \eta)=\frac{1}{2}(\xi+By)^2+\frac{1}{2}\eta^2+x+V(x, y). 
\]

We denote the trapped set for the classical flow in the energy interval $[a,b]$ by $K_{[a,b]}$.
Thus $K_{[a,b]}$ is the set of all $(x_0, y_0, \xi_0, \eta_0)\in T^*\RR^2$ 
such that $a\le p(x_0, y_0, \xi_0, \eta_0)\le b$ and $\sup_{t\in \RR}|(x(t), y(t))|<\infty$, 
where $(x(t), y(t), \xi(t), \eta(t))$ is the solution of the Hamilton equation for 
$p(x, y, \xi, \eta)$ with the initial value $(x_0, y_0, \xi_0, \eta_0)$.

We study the shape resonance model. 
\begin{asm}\label{asm-shape}
We fix $a<b$. We assume $\{(x,y) \in \RR^2|x+V(x, y)\le b\}=\mathcal{G}^{\mathrm{int}}\cup\mathcal{G}^{\mathrm{ext}}$, 
where $\mathcal{G}^{\mathrm{int}}$ is compact, 
$\mathcal{G}^{\mathrm{ext}}$ is closed, and $\mathcal{G}^{\mathrm{int}}\cap\mathcal{G}^{\mathrm{ext}}=\emptyset$. 
Moreover, we assume 
$K_{[a,b]}\cap\{(x, y, \xi, \eta)|(x, y)\in \mathcal{G}^{\mathrm{ext}}\}=\emptyset$.
\end{asm}
This assumption means that there exist potential wells on $\mathcal{G}^{\mathrm{ext}}$ and 
there are no trapped trajectories outside the wells in the energy interval in consideration. 
The latter condition enables us to concentrate on resonances generated by the potential wells, 
which are called shape resonances.  

Our first main theorem is the Weyl-type asymptotics for the magnetic Stark shape resonances.  
\begin{thm}\label{thm-1}
Under Assumption~\ref{asm-res} and Assumption~\ref{asm-shape}, there exists $S>0$ such that 
the numbers of resonances of $P(h)$ in $[a, b]-i[0, e^{-S/h}]$ satisfies 
\[
\#\left(\mathrm{Res}(P(h))\cap ([a, b]-i[0, e^{-S/h}])\right)
=(2\pi h)^{-2} \mathrm{Vol}(K_{[a,b]})+o(h^{-2})
\]
when $h \to 0$.  
\end{thm}
\begin{remark}\label{rem-sharp}
If we assume that $\partial p\not=0$ on 
$p^{-1}(\{a, b\})\cap \{(x, y, \xi, \eta)\,|(x, y)\in \mathcal{G}^{\mathrm{int}}\}$, 
the remainder is improved to $\mathcal{O}(h^{-1})$. 
\end{remark}

We next discuss resonances generated from the bottom of a well. 
\begin{asm}\label{asm-nondege} 
We fix $E\in \RR$. We assume $\{(x,y) \in \RR^2|x+V(x, y)=E\}=\mathcal{G}^{\mathrm{int}}\cup\mathcal{G}^{\mathrm{ext}}$, 
where $\mathcal{G}^{\mathrm{int}}=\{(x_0,y_0)\}$, 
$\mathcal{G}^{\mathrm{ext}}$ is closed, and $\mathcal{G}^{\mathrm{int}}\cap\mathcal{G}^{\mathrm{ext}}=\emptyset$. 
Moreover, we assume that the Hessian matrix 
${\rm Hess\,}V(x_0,y_0)$ is positive definite.
\end{asm}

Let $\lambda_1,\lambda_2>0$ be the two eigenvalues of ${\rm Hess\,}V(x_0,y_0)$. 
Set 
\begin{equation}\label{alphaj}
\alpha_{j}=\sqrt{\frac{B^2+\lambda_1+\lambda_2+(-1)^j
\sqrt{(B^2+\lambda_1+\lambda_2)^2-4\lambda_1\lambda_2}}{2}},\quad j=1,2.
\end{equation}
We give the asymptotic behavior of the resonances of $P(h)$
in the interval
$(E-h^{\delta}, E+h^{\delta})-i[0,e^{-S/h}]$
for any $\delta >0$ and for some $S>0$.

\begin{thm}\label{thm-2}
Assume that Assumption~\ref{asm-res} and Assumption~\ref{asm-nondege} hold, 
and $\alpha_1,\alpha_2$ are $\mathbb{Z}$-independent. 
Then there exist $S>0$ and a real-valued smooth function
$$
F(\tau_1,\tau_2;h)\sim F_0(\tau_1,\tau_2)+F_1(\tau_1,\tau_2)h+\cdots,\quad \tau=(\tau_1, \tau_2)\in\RR^2 
$$
when $h\to 0$ with $F_j={\rm const.}$ for $\tau_1+\tau_2\geq1$, 
$F_0(\tau)=\sum\alpha_j\tau_j+\mathcal{O}(|\tau|^2)$ when $|\tau|\to 0$ 
and $F_0>0$ when $\tau_j\geq0, \tau\not=0$
such that for every fixed $\delta>0$ and for $h>0$ small enough, the real part of the resonances of $P(h)$ in 
$(E-h^{\delta}, E+h^{\delta})-i[0,e^{-S/h}]$ are of the form
$$
E+F\Bigl(\bigl(k_1+\frac{1}{2}\bigr)h,\bigl(k_2+\frac{1}{2}\bigr)h;h\Bigr)+\mathcal{O}(h^{\infty}),\,k_1,k_2\in\mathbb{Z}_{\ge 0}.
$$
\end{thm}

Theorem~\ref{thm-1} and Theorem~\ref{thm-2} 
are proved as corollaries of a one-to-one correspondence of shape resonances and discrete eigenvalues 
of a certain reference operator. 
In the decaying potential case, Helffer and Sj\"{o}strand \cite{HeSjo} and 
Stefanov \cite{S} \cite{S2} established this type of one-to-one correspondence.  
Nakamura, Stefanov and Zworski \cite{NSZ} provided a simplified proof and we follow the strategy of 
\cite{NSZ} to prove Theorem~\ref{thm-1} and Theorem~\ref{thm-2}. 
We also use arguments in Kameoka \cite{K1}, 
where a similar one-to-one correspondence is proved for Stark Hamiltonians. 

This paper is organized as follows.
In Section~2, we define the magnetic Stark resonances by the complex distortion outside a compact set. 
In Section~3, we apply the method in Section~2 to the semiclassical study of resonances for magnetic Stark Hamiltonians. 
We first prove the non-trapping resolvent estimate. 
We then study the shape resonance model and prove the existence of a one-to-one correspondence of 
shape resonances and the eigenvalues of a reference operator. This implies Theorem~\ref{thm-1} and Theorem~\ref{thm-2}.

\section{Exterior complex translation method for magnetic Stark Hamiltonians}\label{sec-2}
\subsection{Definition and basic properties} 
We define resonances of $P$ by the complex translation outside a compact set. 
In this section, we set $h=1$.  
Thus the complex distortion in this section is applicable both to semiclassical and non-semiclassical problems. 
Take $\chi_0\in \Cinf(\RR^2)$ such that $\chi_0(x, y)=1$ when $x^2+y^2<(R_0+1)^2$, where $R_0$ is that in Assumption~\ref{asm-res}. 
We define a vector field 
\[
v(x, y)=(v_1(x,y),v_2(x,y))=(1-\chi_0(x, y), 0). 
\]
We set 
\begin{align*}
\Phi_{\theta}(x, y)=(x, y)+\theta v(x, y)=(x+\theta (1-\chi_0(x, y)), y),
\end{align*}
\begin{align*}
U_{\theta}f(x)&=\det \Phi_{\theta}'(x, y)^{1/2}f(\Phi_{\theta}(x, y))\\
              &=(1-\theta\partial_x \chi_0(x, y))^{1/2}f(x+\theta (1-\chi_0(x, y)), y), 
\end{align*}
which is unitary on $L^2(\RR^2)$ for real $\theta$ with small $|\theta|$. 
Then we define the distorted operator 
\[
P_{\theta}=U_{\theta}P U_{\theta}^{-1}
\]
for real $\theta$ with small $|\theta|$. 
Then $P_{\theta}$ is a differential operator whose coefficients are analytic with respect to $\theta$ 
with $|\IM \theta|<\delta_0$ under Assumption~\ref{asm-res}. 
Thus $P_{\theta}$ for complex $\theta$ is defined as a non-self-adjoint differential operator. 
As an unbounded operator, we define $P_{\theta}$ as the closure of $P_{\theta}$ on $\Cinf_c(\RR^2)$. 
In this section, we confirm that $P_{\theta}$ is analytic with respect to $\theta$ in the operator-theoretic sense 
and the spectrum of $P_{\theta}$ is discrete in $\{z\in \CC|\, \IM z>\IM \theta\}$ 
for $\IM \theta<0$. 
We then show that resonances of $P$ coincide with the discrete eigenvalues of $P_{\theta}$.

\begin{prop}\label{prop2-1}
For $|\IM \theta|<\delta_0$ and small $|\RE \theta|$, $P_{\theta}$ is an analytic family of type(A) and 
$P_{\theta}^*=P_{\bar{\theta}}$.
\end{prop}
\begin{proof}
We first claim that we have 
\begin{equation}\label{perturbation}
\|(P_{\theta}-P_{\theta'})u\|\le C|\theta-\theta'|\|P_{\theta}u\|+C\|u\| 
\end{equation}
for $u\in \Cinf_c(\RR^2)$. 
To prove this, we separate $(P_{\theta}-P_{\theta'})u$ into the the scalar potential term and 
the kinetic term. 
Note that the scalar potential term of $P_{\theta}-P_{\theta'}$ is given by
$(\theta-\theta')v_1(x,y)+V_{\theta}(x, y)-V_{\theta'}(x, y)$. 
This is a bounded function and thus is estimated by $C\|u\|$ when it is applied to $u$. 
The kinetic term of $P_{\theta}-P_{\theta'}$ is a compactly supported second-order differential operator with $\mathcal{O}(|\theta-\theta'|)$ 
coefficients. Thus the elliptic estimate implies that this is estimated by $C|\theta-\theta'|\|P_{\theta}u\|+C\|u\|$ 
when it is applied to $u$. Thus we proved \eqref{perturbation}. 
The equation~\eqref{perturbation} and the Kato-Rellich-type argument implies this lemma 
as in \cite{K1}. 
%as follows. 
\end{proof}

We next discuss the discreteness of the spectrum of $P_{\theta}$. 
\begin{prop}\label{prop2-2}
The spectrum of $P_{\theta}$ is discrete in $\{z\in \CC|\, \IM z>\IM \theta\}$ 
for $\IM \theta<0$ and small $|\RE \theta|$. 
\end{prop}

To prove Proposition~\ref{prop2-2}, we first modify $P_{\theta}$ on a compact set as follows.
Take $\chi\in \Cinf_c(\RR^2)$ such that $\chi=1$ near the origin.  
We take $R>0$ and set 
\[
P_{\theta, R}=P_{\theta}+R\chi(x/R, y/R).
\]

\begin{lemma}\label{lem-inv}
For any complex number $z$ with $\IM z>\IM \theta$, there exists $R>0$ such that 
$z \in \CC\setminus \sigma(P_{\theta, R})$ and $z \in \CC\setminus \sigma(P_{\theta}+R)$.
\end{lemma}

\begin{proof}[Proof of Proposition~\ref{prop2-2} assuming Lemma~\ref{lem-inv}]
Take $z$ and $R$ as in Lemma~\ref{lem-inv}. We write $\chi^R(x, y)=\chi(x/R, y/R)$ and 
take $\tilde{\chi}\in \Cinf(\RR^2)$ such that 
$\tilde{\chi}=1$ on $\supp \chi^R$.  
Then the elliptic estimate implies 
\[
\| \Delta \chi^R u\|\lesssim \|\tilde{\chi}(P_{\theta, R}-z)u\|+\|\tilde{\chi} u\|
\]
for $u\in D(P_{\theta, R})$. This implies 
\[
\|\Delta\chi^M (P_{\theta, R}-z)^{-1} u\|\lesssim \|\tilde{\chi}u\|+\|\tilde{\chi} (P_{\theta, R}-z)^{-1}u\|
\]
for $u\in L^2(\RR^2)$. 
Since $(P_{\theta, R}-z)^{-1}$ is a bounded operator by Lemma~\ref{lem-inv}, the Rellich theorem implies that 
$\chi^R (P_{\theta, R}-z)^{-1}$ is a compact operator. 

We note that 
\[
P_{\theta}-z=(1-R\chi^R(P_{\theta, R}-z)^{-1})(P_{\theta, R}-z). 
\]
Since $\chi^R (P_{\theta, R}-z)^{-1}$ is a compact operator, this is a Fredholm operator with index zero. 
Then $P_{\theta}-z$, $\IM z>\IM \theta$ is an analytic family of Fredholm operators with index zero. 
This and the second claim in Lemma~\ref{lem-inv} imply that $(P_{\theta}-z)^{-1}$ is meromorphic, which completes the proof. 
\end{proof} 

In the case of the complex translation on the whole plane, 
$z \in \CC\setminus \sigma(P_{\theta}-i\chi^R)$ and $z \in \CC\setminus \sigma(P_{\theta}-iR)$ are easily proved 
by considering $\IM (u, (P_{\theta}-i\chi^R-z)u)$ and $\IM (u, (P_{\theta}-iR-z)u)$. 
Thus the proof of Proposition~\ref{prop2-2} is easy in this case. 
Since we use the exterior complex translation, the negative part of the second-order part of $P_{\theta}$ is not 
negative semi-definite. The proof of Lemma~\ref{lem-inv} does not seem easy and postponed to the next subsection.

We finally characterize resonances of $P$ as discrete eigenvalues of $P_{\theta}$. 
\begin{prop}\label{prop2-3}
The resonances of $P$ in $\{z\in \CC|\, \IM z>\IM \theta\}$ 
coincides with discrete eigenvalues of $P_{\theta}$ in that region including multiplicities 
for $-\delta_0 <\IM \theta<0$ and small $|\RE \theta|$. 
\end{prop}
\begin{proof}
Once we established Proposition~\ref{prop2-1} and Proposition~\ref{prop2-2}, 
the Proposition~\ref{prop2-3} follows from well-known arguments in resonance theory as follows.  
Take any nonzero $\chi_1, \chi_2\in \Cinf_c(\RR^2)$. 
We then distort $P$ outside the support of $\chi_1, \chi_2$ and construct $P_{\theta}$. 
Then 
\begin{equation}\label{cutoff-res}
\chi_1 R_{+}(z)\chi_2=\chi_1 (z-P_{\theta})^{-1}\chi_2
\end{equation} 
for $\IM z>0$ and real $\theta$ with small $|\theta|$. 
By Proposition~\ref{prop2-1}, the equation\eqref{cutoff-res} is true for 
$-\delta_0<\IM \theta<0$ and small $|\RE \theta|$ by the analytic continuation. 
By Proposition~\ref{prop2-2}, the left hand side of the equation~\eqref{cutoff-res} has a meromorphic continuation 
to $\{z\in \CC|\, \IM z>\IM \theta\}$ and the equation~\eqref{cutoff-res} is true for $\IM z>\IM \theta$. 
Then the multiplicity $m_z$ of resonance $z$ is given by 
\[
m_z=\mathrm{rank}\frac{1}{2\pi i}\chi_1 \oint_{|\zeta-z|=c} (z-P_{\theta})^{-1}  d\zeta \chi_2
=\mathrm{rank}\frac{1}{2\pi i}\oint_{|\zeta-z|=c} (z-P_{\theta})^{-1} d\zeta, 
\]
where the second equation follows from the unique continuation principle. 
Thus the multiplicity $m_z$ of resonance $z$ coincides with the algebraic multiplicity of $z$ as an eigenvalue of $P_{\theta}$.  
This in particular implies that $m_z$ is independent of $\chi_1, \chi_2$. 
\end{proof}

\subsection{Proof of Lemma~\ref{lem-inv}}
We use the following lemma to prove Lemma~\ref{lem-inv}. 
Take any $\chi_2 \in \Cinf_c(\RR^2)$ and set $\chi_2^R(x, y)=\chi_2(x/R, y/R)$. 

\begin{lemma}\label{lem-elliptic}
There exist $C>0$ and $R_0>0$ such that 
\begin{equation}\label{eq-1}
\|\chi_2^R (D_x+By)u\|+\|\chi_2^R D_y u\|+\|\chi_2^R u\|\le C\|(P_{\theta, R}-z)u\|+CR^{-1}\|u\|
\end{equation}
for $u\in \Cinf_c(\RR^2)$ and $R>R_0$. 
\end{lemma}

\begin{remark}
In fact, the right hand side in \eqref{eq-1} is improved to $CR^{-1/2}\|(P_{\theta, R}-z)u\|$ by some 
modifications of the proof. 
\end{remark}

\begin{proof}[Proof of Lemma~\ref{lem-inv} assuming Lemma~\ref{lem-elliptic}]
We need to show that $(P_{\theta, R}-z)^{-1}$ is a bounded operator on $L^2(\RR^2)$ for $R\gg 1$. 
The same proof with $R\chi^R$ replaced by $R$ shows that $(P_{\theta}+R-z)^{-1}$ is a bounded operator on $L^2(\RR^2)$ for $R\gg 1$. 
We take $\chi_1\in \Cinf_c(\RR^2)$ such that $\chi_1(x, y)=1$ near $(0, 0)$ and that 
$\supp \chi_1$ is sufficiently small so that $\chi=1$ on $\supp \chi_1$.  
We also set $\chi_1^R(x, y)=\chi_1(x/R, y/R)$. 

We estimate $\|(P_{\theta, R}-z)u\|$ from below for $u\in \Cinf_c(\RR^2)$.  
We see that 
\[
\RE(\chi_1^R u, (P_{\theta, R}-z)\chi_1^R u)\gtrsim\|\chi_1^R u\|^2. 
\]
This follows from the elliptic estimate since the second order part of $P_{\theta}$ is positive definite since $|\RE \theta|\ll 1$ 
and $R\chi^R(x)>R$ when $\chi_1^R(x)\not=0$. 

We also see that 
\[
-\IM((1-\chi_1^R) u, (P_{\theta, R}-z)(1-\chi_1^R) u)\gtrsim\|(1-\chi_1^R) u\|^2. 
\]
This follows from the $\theta v_1(x, y)$ term of the complex distortion of the Stark potential 
since $\IM\theta v(x, y)=\IM \theta$ when $(1-\chi_1^R(x))\not=0$. 

We then have
\begin{align*}
\|u\|&\le  \|\chi_1^R u\|+\|(1-\chi_1^R) u\|\\
     & \lesssim \|(P_{\theta, R}-z)\chi_1^R u\|+\|(P_{\theta, R}-z)(1-\chi_1^R) u\|\\ 
	 &\lesssim \|(P_{\theta, R}-z) u\|+\|[P_{\theta, R}, \chi_1^R] u\|. 
\end{align*}
Lemma~\ref{lem-elliptic} implies 
\[
\|[P_{\theta, R}, \chi_1^R] u\| \lesssim R^{-1}\|(P_{\theta, R}-z)u\|+R^{-2}\|u\|. 
\]
We conclude that 
\[
\|u\|\le C\|(P_{\theta, R}-z)u\|
\]
for $R\gg 1$. By approximation, this is valid for all $u$ in the domain of $P_{\theta, R}$. 
Since we have the same estimate for the adjoint operator, we see that 
$\|(P_{\theta, R}-z)^{-1}\|_{L^2\to L^2}\le C$ for $R\gg 1$.

\end{proof}

It remains to prove Lemma~\ref{lem-elliptic}. 
For that, we use the theory of pseudodifferential operators (see~Zworski~\cite{Z}). 
we take $w \in  \Cinf(\RR^2; \RR_{\ge 1})$ depending only on $x$ and 
$w=|x|$ for $x \le -2$ and $w=1$ for $x \ge -1$. 
We fix $\theta$ with $-\delta_0<\IM \theta<0$ and $|\RE \theta|$ small. 
We also fix $z\in \CC$ with $\IM z>\IM \theta$. 
We introduce an auxiliary operator
\[
A_R=P_{\theta, R}-z+2w(x/R).
\]
We define the order function $m$ by 
\[
m(x, y, \xi, \eta)=((\xi+By)^2+\eta^2+1)^{1/2}.
\]

\begin{lemma}\label{lem-symbol}
There exists $R_1>0$ such that $A_R^{-1}\in \mathrm{Op}S(m^{-2})$ uniformly for $R>R_1$. 
Namely, each seminorm of the symbol of $A_R^{-1}$ in $S(m^{-2})$ is bounded for $R>R_1$. 
\end{lemma}
\begin{proof}
We define the symbol $a$ by the formula $A=a(x, y, D_x, D_y; R)$. 
We estimate $\partial^{\alpha}_{x, y} \partial^{\beta}_{\xi, \eta}a(x, y, \xi, \eta; R)^{-1}$ 
for $R \gg 1$. 

We first see that  
\begin{equation}\label{estimate-1}
\left|\frac{m^k}
{a(x, y, \xi, \eta; R)}\right|\lesssim R^{-1+k/2}
\end{equation}
for $0\le k \le 2$. 
This follows for $m/R^{1/2}\gg 1$ since $|a|\gtrsim m^2$, 
and for $m/R^{1/2}\lesssim 1$ since $|a|\ge \RE a \gtrsim R$. 
The inequality~\eqref{estimate-1} implies that
\begin{align*}
\left|\frac{\partial_{\xi, \eta} a(x, y, \xi, \eta; R)}
{a(x, y, \xi, \eta; R)}\right|
\lesssim \left|\frac{m}
{a(x, y, \xi, \eta; R)}\right| \lesssim R^{-1/2}. 
\end{align*}
The inequality~\eqref{estimate-1} also implies that
\begin{align*}
\left|\frac{\partial_{x, y} a(x, y, \xi, \eta; R)}
{a(x, y, \xi, \eta; R)}\right|
\lesssim \frac{m^2}
{|a(x, y, \xi, \eta; R)|}\lesssim 1. 
\end{align*}

When we estimate $\partial^{\alpha}_{x, y} \partial^{\beta}_{\xi, \eta}a(x, y, \xi, \eta)^{-1}$, 
taking the derivatives of the numerators does not make the estimates worse and taking the derivatives of the denominators
multiplies $(\partial_{x, y} a)/a$ or $(\partial_{\xi, \eta}a)/a$. 
Thus $|(\partial_{x, y} a)/a|, |(\partial_{\xi, \eta}a)/a| \lesssim 1$ and the inequality \eqref{estimate-1} imply 
that $a(x, y, \xi, \eta; R)^{-1}$ is $\mathcal{O}(R^{-1+k/2})$ in $S(m^{-k})$ when $R\to \infty$. 
This implies that $\partial_{x, y} a(x, y, \xi, \eta; R)^{-1}$ is $\mathcal{O}(R^{-1/2})$ in $S(m^{-1})$. 
We also see that $\partial_{\xi, \eta} a(x, y, \xi, \eta; R)^{-1}$ is $\mathcal{O}(R^{-1/2})$ in $S(m^{-2})$.  
This can be seen from 
\[
\left|\frac{m^2}
{a(x, y, \xi, \eta; R)}\cdot \frac{m}
{a(x, y, \xi, \eta; R)}\right| \lesssim R^{-1/2} 
\]
and $|(\partial_{x, y} a)/a|, |(\partial_{\xi, \eta}a)/a| \lesssim 1$. 

Note that $\partial_{x, y} a \in S(m^2)$ and $\partial_{\xi, \eta}a \in S(m)$ uniformly for $R\gg1$ since 
derivatives of $R\chi^R$ are uniformly bounded for $R\gg 1$. 
Then the symbol calculus for pseudodifferential operators implies that 
the symbols of 
\[
a(x, y, D_x, D_y; R)a^{-1}(x, y, D_x, D_y; R)-1
\]
and 
\[
a^{-1}(x, y, D_x, D_y; R)a(x, y, D_x, D_y; R)-1
\]
are $\mathcal{O}(R^{-1/2})$ in $S(1)$ (see \cite[Chapter~4]{Z}). 
Thus the Neumann series argument and Beals's theorem complete the proof.
\end{proof}

\begin{proof}[Proof of Lemma~\ref{lem-elliptic}]
Replacing $R$ with $cR$ for $0<c\ll 1$, we may assume that $\supp \chi_2$ is sufficiently small. 
Recall that $m=((\xi+By)^2+\eta^2+1)^{1/2}$. 
We introduce the magnetic Sobolev norm by $\|u\|_{H_{A, h}^k}=\|m_{\lambda}^k(x, y, D_x, D_y)u\|$, where 
$m_{\lambda}=((\xi+By)^2+\eta^2+\lambda)^{1/2}$ and $\lambda=\lambda_k>1$ is a sufficiently large constant such that 
$m_{\lambda}^k(x, y, D_x, D_y)^{-1}\in \mathrm{OP}S(m^{-k})$. 
Here $A=-Bydx$ is the vector potential for our Hamiltonian.
Then Lemma~\ref{lem-symbol} implies that the left hand side of Lemma~\ref{lem-elliptic} is bounded by 
\begin{align*}
C\|\chi_2^R u\|_{H_{A, h}^1}&\lesssim \|A_R \chi_2^R u\|_{H_{A, h}^{-1}}\\
&= \|(P_{\theta, R}-z)\chi_2^R u\|_{H_{A, h}^{-1}} \\
&\le \|\chi_2^R(P_{\theta, R}-z)u\|_{H_{A, h}^{-1}}+\|[P_{\theta, R}, \chi_2^R]u\|_{H_{A, h}^{-1}} \\
&\lesssim \|(P_{\theta, R}-z)u\|+R^{-1}\|u\|, 
\end{align*}
where we used the fact that $A_R=P_{\theta, R}-z$ on $\supp \chi_2^R$ since $\supp \chi_2$  is small. 
\end{proof}

\section{Semiclassical estimates of resonances}
\subsection{Non-trapping estimate}
Recall that $K_{[a, b]}$ is the trapped set for the classical Hamilton flow in the energy interval $[a, b]$. 
We first prove the non-trapping estimates of resonances and cutoff resolvent for magnetic Stark Hamiltonians. 
In this subsection, we fix the vector field $v(x)$ in the construction of $P_{\theta}(h)$. 
Namely, we define $U_{\theta}$ as in Section~2 and set $P_{\theta}(h)=U_{\theta}P(h)U_{\theta}^{-1}$ for real $\theta$ with 
small $|\theta|$. Then $P_{\theta}(h)$ for complex $\theta$ is defined through analytic continuation with respect to $\theta$. 
Although $v(x)$ is independent of $h$, the distortion parameter $\theta$ 
is dependent on $h$ in the following proposition. 
The following proposition is due to Martinez \cite{M} for the decaying potential case.
\begin{prop}\label{prop-nontrapping}
Suppose that Assumption~\ref{asm-res} holds and $K_{[a,b]}=\emptyset$ for some $a<b$. 
Then for any $M>0$ there exists $\widetilde{M}>0$ and $C>0$ 
such that for small $h>0$ and $z \in [a ,b]+i[-Mh\log h^{-1}, \infty)$, 
\[
\|(P_{\theta}(h)-z)^{-1}\|\le h^{-C} 
\]
where $(\IM z)_-=\max\{-\IM z, 0\}$ and $\theta=-i\widetilde{M}h\log h^{-1}$. 
\end{prop}
\begin{proof}
We follow the arguments of Sj\"{o}stand and Zworski \cite[Theorem 1]{SZ} 
while we need a modification due to the non-elliptic aspect of magnetic Stark Hamiltonians.
Take a sufficiently large $R>0$. 
Then the non-trapping assumption enables us to construct a function $G\in \Cinf_c(T^*\RR^2)$ 
such that $\{p, G\}\ge 1$ on $p^{-1}([\widetilde{a}, \widetilde{b}])\cap 
\{(x, y, \xi, \eta)|\,|(x, y)|<R\}$ 
for some $\widetilde{a}<a<b<\widetilde{b}$. 
Moreover we may assume that the negative part of $\{p, G\}$ is sufficiently small on $T^*\RR^2$.  
Here $\{p, G\}$ is the Poisson bracket of $p$ and $G$ 
(see \cite[subsection~4.2]{SZ}). 
We set $P_{\theta, \ve}(h)=e^{-\ve G^W/h}P_{\theta}(h)e^{\ve G^W/h}$, 
where $\ve=M'h\log h^{-1}$ with $M'\gg \widetilde{M}$. 
We consider $z$ with $a\le \RE z \le b$ and $(\IM z)_-\ll \ve$.  

We take microlocal cutoffs $\Psi_1$, $\Psi_2$ and $\Psi_3$ near 
$\{x \ge R_1\}\cup \{|x|<R_1, |y|<R', p(x, y, \xi, \eta) \not \in [\widetilde{a}, \widetilde{b}]\}, 
\{|x|<R_1, |y|<R',  p(x, y, \xi, \eta) \in [\widetilde{a}, \widetilde{b}]\}$ and 
$\{x<-R_1\}\cup \{|x|<R_1, |y|>R'\}$ respectively, where $1\ll R_1 \ll R' \ll R$.  
We take $u\in \Cinf_c(\RR^2)$. 
The elliptic estimate implies 
$\|(P_{\theta, \ve}-z)\Psi_1 u\|\ge c\|\Psi_1 u\|-\mathcal{O}(h^{\infty})\|u\|$ 
since $R_1\gg 1$ due to the Stark potential. 
By the condition that $\{p, G\}\ge 1$ on $p^{-1}([\widetilde{a}, \widetilde{b}])\cap \{|(x, y)|<R\}$ and 
$M'\gg \widetilde{M}$, the imaginary part of the symbol of $(P_{\theta, \ve}$ is larger than 
$c_0\ve$ for some $c_0>0$ on the support of the symbol of $\Psi_2 $. 
Then the sharp G\r{a}rding inequality implies that there exists $c>0$ such that 
$\|(P_{\theta, \ve}-z)\Psi_2 u\|
\ge c\ve\|\Psi_2 u\|-\mathcal{O}(h^{\infty})\|u\|$ 
for $(\IM z)_-\ll \ve$.   
Since the semiclassical principal symbol of $P_{\theta, \ve}$ is not globally elliptic due to the Stark potential, 
we estimate $\Psi_3 u$ by considering the quadratic form. 
Then we see that
$-\IM (\Psi_3 u, (P_{\theta, \ve}-z)\Psi_3 u) \ge c|\IM\theta| \|\Psi_3 u\|^2$ 
and hence we have 
$\|(P_{\theta, \ve}-z)\Psi_3 u\| \ge c|\IM\theta| \|\Psi_3 u\|$
for $(\IM z)_-\ll |\IM\theta|$ 
since the negative part of $\{p, G\}$ is sufficiently small on $T^*\RR^2$. 
 
Thus 
\begin{align*}
\|u\|
&\le C\ve^{-1}\sum_{j=1}^3 \|(P_{\theta, \ve}-z)\Psi_j u\| \\
&\le C\ve^{-1} \|(P_{\theta, \ve}-z)u\|+C\ve^{-1} \sum_{j=1}^3 \|[P_{\theta, \ve}\Psi_j] u\| \\
&\le C\ve^{-1} \|(P_{\theta, \ve}-z)u\|+Ch/{\ve}(\||(P_{\theta, \ve}-z)u\|+\|u\|).
\end{align*}
Since $\ve=M'h\log h^{-1}$, we substitute $Ch/{\ve}\|u\|<1/2\|u\|$ and 
we obtain $\|(P_{\theta, \ve}-z)u\|\ge c\ve \| u\|$.
Since $\|e^{\pm\ve G^W/h}\|\lesssim h^{-C_0}$ for some $C_0>0$, this completes the proof. 
\end{proof}

\subsection{Shape resonance model}
In this subsection, we discuss the shape resonances for the magnetic Stark Hamiltonian generated by potential wells. 
Recall that $p(x, y, \xi, \eta)=\frac{1}{2}(\xi+By)^2+\frac{1}{2}\eta^2+x+V(x, y)$. 
and $K_{[a,b]}$ is the trapped set in the energy interval $[a,b]$. 
Throughout this subsection, we assume Assumption~\ref{asm-res} and Assumption~\ref{asm-shape}. 
Assumption~\ref{asm-shape} implies 
\[
K_{[a,b]}=\{(x, y,\xi, \eta)|(x, y)\in \mathcal{G}^{\mathrm{int}}, a\le p(x, y, \xi, \eta)\le b\}.
\] 
We can take small $\delta>0$ such that Assumption~\ref{asm-shape} holds true with $[a,b]$ replaced by $[a-\delta, b+\delta]$, 
and there exists a cutoff function $\chi_0$ near $\mathcal{G}^{\mathrm{int}}$ such that 
\[
\supp \partial \chi_0 \Subset \{(x, y)\in \RR^2|x+V(x, y)>b+2\delta\}.
\] 
We fix such $\delta>0$ and $\chi_0$. 
Complex distorted operators in this subsection are constructed outside $\supp \chi_0$. 
Let $V^{\mathrm{ext}}(x, y)$ be a potential obtained by filling up the well: 
$V^{\mathrm{ext}}(x, y)=V(x, y)+x$ near $\supp (1-\chi_0)$ and 
$V^{\mathrm{ext}}(x, y)>b+2\delta$ near $\mathcal{G}^{\mathrm{int}}$, 
and $P^{\mathrm{ext}}_{\theta}(h)$ be the corresponding distorted operator.
Let $V^{\mathrm{int}}(x, y)$ be a potential flattened outside the well: 
$V^{\mathrm{int}}(x)=V(x, y)+x$ near $\supp \chi_0$ 
and $V^{\mathrm{int}}(x, y)=b+2\delta$ near $\mathcal{G}^{\mathrm{ext}}$, 
and 
\[
P^{\mathrm{int}}=\frac{1}{2}(hD_{x, y}+By)^2+\frac{1}{2}(hD_y)^2+V^{\mathrm{int}}.
\]
Then $P^{\mathrm{int}}$ is our reference operator and we show that the resonances of $P$ are approximated by discrete eigenvalues of 
$P^{\mathrm{int}}$. 
In the following we set $\alpha(h)=h^C$ and $\gamma(h)=Mh\log h^{-1}$
for a fixed $M>0$ and fixed large $C>0$, $\widetilde{M}>0$. 
Then 
$\|(P^{\mathrm{ext}}_{\theta}(h)-z)^{-1}\|=\mathcal{O}(\alpha(h)^{-1})$ 
for $a-\delta \le \RE z \le b+\delta$, $\IM z \ge -\gamma(h)$ 
and $\theta =-i\widetilde{M}h \log h^{-1}$ by Proposition~\ref{prop-nontrapping}.  
% \begin{remark}
% The results in this subsection remain true 
% if we replace the non-trapping condition outside the well by a resolvent assumption as follows: 
% there exist functions $\alpha (h)$, $\gamma(h)$ and real numbers $a<b$ with 
% $\alpha(h), \gamma(h)>e^{-S/h}$ for any $S>0$ for small $h>0$ such that 
% $\|(P^{\mathrm{ext}}_{\theta}(h)-z)^{-1}\|=\mathcal{O}(\alpha(h)^{-1})$ 
% for $a-\delta \le \RE z \le b+\delta$ and $\IM z \ge -\gamma(h)$. 
% \end{remark}
The basic estimate in this subsection is the following magnetic Agmon estimate 
based on the magnetic Sobolev space. 
Recall that the magnetic Sobolev norm was introduced in 
the proof of Lemma~\ref{lem-elliptic}. 
The proof for the zero magnetic field case in \cite[Chapter~7]{Z}) is easily extended to this case.   
\begin{lemma}\label{lemma-Agmon}
For any $\chi\in \Cinf_c(\RR^2)$ with
$\supp \chi \subset \{(x, y)\in \RR^2|V(x, y)+x>b+2\delta\}$ and $C_0>0$, 
there exists $S_0>0$ and $C>0$ such that for 
any $z\in [b-C_0, b+\delta]+i[-C_0,C_0]$ and small $h>0$, 
\[ \|\chi u\|_{H_{A, h}^2} \le e^{-S_0/h}\|u\|+C\|(P-z)u\|. \]
\end{lemma}
This is also valid for $P_{\theta}$ if $\supp \chi$ is away from the region of deformation 
in the definition of $P_{\theta}$. 
In the following we fix $S_0$ such that  
Lemma~\ref{lemma-Agmon} holds true where $\chi$ is a cutoff near $\supp \partial \chi_0$. 
Moreover we may assume Lemma~\ref{lemma-Agmon} with $P$ replaced by $P^{\mathrm{int}}$ holds true  
where $\chi$ is a cutoff near $\supp (1-\chi_0)$.

The crucial estimate in our arguments is the following form of resolvent estimate. 
This type of estimate was proved for Stark Hamiltonians in \cite[Proposition~4.1]{K1}. 
\begin{prop}\label{prop-resolvent}
For small $h>0$,  
\[\|(1-\chi_0)(P_{\theta}-z)^{-1}\|\le C\alpha(h)^{-1}, \|\chi_0(P_{\theta}-z)^{-1}\|
\le C\mathrm{dist}(z, \sigma(P^{\mathrm{int}}))^{-1}\]
if $a-\delta \le \RE z \le b+\delta$, $\IM z \ge -\gamma(h)$ 
and $\mathrm{dist}(z, \sigma(P^{\mathrm{int}}))\ge e^{-S_0/h}$.
\end{prop}
\begin{proof}
We recall the proof in \cite{K1} for convenience.
The Agmon estimate implies 
\begin{align*}
\|(1-\chi_0)(P_{\theta}-z)^{-1}\|&
=\|(P_{\theta}^{\mathrm{ext}}-z)^{-1}(P_{\theta}^{\mathrm{ext}}-z)(1-\chi_0)(P_{\theta}-z)^{-1}\|\\
&\le\alpha(h)^{-1}\|(P_{\theta}-z)(1-\chi_0)(P_{\theta}-z)^{-1}\|\\
&\le\alpha(h)^{-1}(1+\|[P_{\theta}, \chi_0](P_{\theta}-z)^{-1}\|)\\ 
&\le C\alpha(h)^{-1}(1+e^{-S_0/{h}}\|(P_{\theta}-z)^{-1}\|)\\
&\le C\alpha(h)^{-1}(1+e^{-S_0/{h}}\|\chi_0(P_{\theta}-z)^{-1}\|).
\end{align*}
The last inequality follows by subtracting $C\alpha(h)^{-1}e^{-S_0/{h}}\|(1-\chi_0)(P_{\theta}-z)^{-1}\|
\le \frac{1}{2}\|(1-\chi_0)(P_{\theta}-z)^{-1}\|$ from both sides for small $h>0$.
The Agmon estimate also implies 
\begin{align*}
\|\chi_0(P_{\theta}-z)^{-1}\|&= \|(P^{\mathrm{int}}-z)^{-1}(P^{\mathrm{int}}-z)\chi_0(P_{\theta}-z)^{-1}\| \\
&\le \mathrm{dist}(z, \sigma(P^{\mathrm{int}}))^{-1}\|(P_{\theta}-z)\chi_0(P_{\theta}-z)^{-1}\| \\
&\le \mathrm{dist}(z, \sigma(P^{\mathrm{int}}))^{-1}(1+\|[P_{\theta},\chi_0](P_{\theta}-z)^{-1}\|) \\
&\le C\mathrm{dist}(z, \sigma(P^{\mathrm{int}}))^{-1}(1+h e^{-S_0/{h}}\|(P_{\theta}-z)^{-1}\|)\\
&\le C\mathrm{dist}(z, \sigma(P^{\mathrm{int}}))^{-1}(1+h e^{-S_0/{h}}\|(1-\chi_0)(P_{\theta}-z)^{-1}\|).
\end{align*}
The last inequality follows by subtracting
$Ch\mathrm{dist}(z, \sigma(P^{\mathrm{int}}))^{-1}e^{-S_0/{h}}\|\chi_0(P_{\theta}-z)^{-1}\|\le 
Ch\|\chi_0(P_{\theta}-z)^{-1}\|$ from both sides for small $h>0$.
We substituting the left hand side of each inequality for the right hand side of the other inequality. 
Then we obtain the claimed inequalities by subtracting the small remainder from both sides. 
\end{proof}

This proposition shows that there exists a gap between shape resonances and the other resonances. 
Namely, we have  
\[\mathrm{Res}(P(h))\cap ([a-\delta, b+\delta]-i[e^{-S_0/h}, \gamma(h)])
=\emptyset\mspace{10mu}\text{for small}\mspace{7mu} h>0.\]

To prove Theorem~\ref{thm-1}, we decompose resonances into clusters following earlier 
works by Stefanov \cite{S2} and Nakamura, Stefanov and Zworski \cite{NSZ}.
\begin{lemma}\label{lemma-cluster}
For small $h>0$, there exist $a_j(h)<b_j(h)<a_{j+1}(h)$ such that 
\[\left(\mathrm{Res}(P(h))\cup 
\sigma(P^{\mathrm{int}})\right)\cap ([a-\frac{\delta}{2}, b+\frac{\delta}{2}]-i[0, e^{-S_0/h}])
\subset \cup_{j=1}^{J(h)}\Omega_j(h),\]
where $\Omega_j(h)=[a_j(h), b_j(h)]-i[0, e^{-S_0/h}]$, 
$b_j-a_j\le Ch^{-2}e^{-S_0/h}$, $a_{j+1}-b_j\ge 2e^{-S_0/h}$, 
$a_1\in(a-\frac{2}{3}\delta, a-\frac{1}{3}\delta)$, $b_{J(h)}\in (b+\frac{1}{3}\delta,b+\frac{2}{3}\delta)$ and 
$\mathrm{Res}(P)\cap(([a_1-ch^2, a_1]-i[0, e^{-S_0/h}])
\cup ([b_{J(h)}, b_{J(h)}+ch^2]-i[0, e^{-S_0/h}]))
=\emptyset$.
Moreover, 
\[\|(1-\chi_0)(P_{\theta}-z)^{-1}\|\le C\alpha(h)^{-1},\mspace{7mu} z\in \partial \widetilde{\Omega}_j(h),\]
where $\widetilde{\Omega}_j(h)=[a_j(h)-e^{-S_0/h}, 
b_j(h)+e^{-S_0/h}]+i[-2e^{-S_0/h}, e^{-S_0/h}]$.
\end{lemma}
\begin{proof}
The first statement follows from the fact that 
$\# (\sigma(P^{\mathrm{int}})\cap [a-\delta, b+\delta])=\mathcal{O}(h^{-2})$ and 
Proposition~\ref{prop-resolvent} by elementary arguments (see \cite{NSZ}). 
The second statement follows from Proposition~\ref{prop-resolvent}.
\end{proof}

Set $ \Pi_j^{\theta}=\frac{1}{2\pi i}\int_{\partial \widetilde{\Omega}_j} (z-P_{\theta})^{-1}dz$ and 
$ \Pi_j^{\mathrm{int}}=\frac{1}{2\pi i}\int_{\partial \widetilde{\Omega}_j} (z-P^{\mathrm{int}})^{-1}dz$. 

\begin{prop}\label{prop-projection}
For any $0<S<S_0$, 
\[ 
\Pi_j^{\theta}= \Pi_j^{\mathrm{int}}+\mathcal{O}(e^{-S/h})\,\,\,\,\,\mathrm{as} \,\,\,h\to 0
\]

\end{prop}

\begin{proof}
Since $z-P^{\mathrm{int}}$ is elliptic near $\supp (P_{\theta}-P^{\mathrm{int}})$, we have 
\begin{align*}
\|&(P_{\theta}-P^{\mathrm{int}})(z-P^{\mathrm{int}})^{-1}\chi_0\|\\
&\le C\|(z-P^{\mathrm{int}})(P_{\theta}-P^{\mathrm{int}})(z-P^{\mathrm{int}})^{-1}\chi_0\|_{L^2\to H_m^{-2}} \\
&= C \|[P^{\mathrm{int}}, P_{\theta}-P^{\mathrm{int}}](z-P^{\mathrm{int}})^{-1}\chi_0\|_{L^2\to H_m^{-2}} \\
&\le C+C e^{-S_0/h}\|(z-P^{\mathrm{int}})^{-1}\|
\end{align*}
where the last inequality follows from the Agmon estimate for $P^{\mathrm{int}}$ 
and the fact that $[P^{\mathrm{int}}, P_{\theta}-P^{\mathrm{int}}]$ has bounded coefficients. 
Then 
\begin{align*}
\|(P_{\theta}-P^{\mathrm{int}})(z-P^{\mathrm{int}})^{-1}\chi_0\|\le C
\end{align*}
when $\mathrm{dist}(z, \sigma(P^{\mathrm{int}}))\ge e^{-S_0/h}$. 

Since $\supp \chi_0\cap \supp (P_{\theta}-P^{\mathrm{int}})=\emptyset$, we have
\begin{align*}
 \Pi_j^{\theta}-\Pi_j^{\mathrm{int}}
=\frac{1}{2\pi i}\int_{\partial \widetilde{\Omega}_j}
(z-P_{\theta})^{-1}(1-\chi_0)(P_{\theta}-P^{\mathrm{int}})(z-P^{\mathrm{int}})^{-1}dz. 
\end{align*}
This and Lemma~\ref{lemma-cluster} imply 
\[
\| \left(\Pi_j^{\theta}- \Pi_j^{\mathrm{int}}\right) \chi_0\|\le C|\partial \widetilde{\Omega}_j| \alpha(h)^{-1} 
=\mathcal{O}(e^{-S/h}).
\] 
Finally, we have $\|\Pi_j^{\theta}(1-\chi_0)\|\le C|\partial \widetilde{\Omega}_j| \alpha(h)^{-1} 
=\mathcal{O}(e^{-S/h})$ by Lemma~\ref{lemma-cluster}, 
and $\|(1-\chi_0)\Pi_j^{\mathrm{int}}\|\le Ch^{-2}e^{-S_0/h}
=\mathcal{O}(e^{-S/h})$ by the Agmon estimate.
\end{proof}

\begin{proof}[Proof of Theorem~\ref{thm-1}]
Proposition~\ref{prop-projection} implies that 
$\mathrm{rank}\,\Pi_j^{\theta}=\mathrm{rank}\,\Pi_j^{\mathrm{int}}$ for small $h>0$.
Thus the Weyl law on eigenvalues of $P^{\mathrm{int}}$ completes the proof.
\end{proof}
Remark~\ref{rem-sharp} on the sharp Weyl asymptotics
 also follows from the sharp remainder version of the Weyl law on discrete eigenvalues 
(see Dimassi, Sj\"{o}strand \cite[Section~10]{DS}). 

To prove Theorem~\ref{thm-2}, 
it is enough to study the asymptotic distribution of eigenvalues of $P^{\mathrm{int}}$ near $E$
since we have $\mathrm{rank}\,\Pi_j^{\theta}=\mathrm{rank}\,\Pi_j^{\mathrm{int}}$ for small $h>0$. 
 While Helffer and Sj\"ostrand \cite{HeSjo2} study the distribution of eigenvalues of magnetic 
Schr\"{o}dinger operators near potential minimum, we apply the result from Sj\"ostrand \cite{Sj}, 
which is based on the Birkhoff normal form for pseudodifferential operators. 

\begin{proof}[Proof of Theorem~\ref{thm-2}]
Without loss of generality, we may assume that $(x_0, y_0)=(0, 0)$. 
Moreover a gauge transform and a linear change of variables of $x$, $y$ show that 
we may assume that 
the semiclassical principal symbol of $P^{\mathrm{int}}$ is 
\[
\frac{1}{2}(\xi+By)^2+\frac{1}{2}\eta^2+x+\frac{1}{2}(\lambda_1 x^2+\lambda_2 y^2)+\mathcal{O}(|(x, y)|^3) 
\]
near $(0, 0)$. 
Note that quadratic forms are classified under linear canonical transform (see Matsumoto and Ueki \cite{MU1}). 
Then the quadratic part of this symbol is equivalent to 
\begin{align*}
\frac{\alpha_1}{2}(x^2+\xi^2)+\frac{\alpha_2}{2}(y^2+\eta^2), 
\end{align*}
where $\alpha_1,\alpha_2$ are given by \eqref{alphaj} in Section~\ref{sec-1}. 
Then the following theorem follows from \cite{Sj} (see also Dimassi and Sj\"{o}strand \cite[Chapter~14]{DS}). 
There exists a real-valued smooth function
$$
F(\tau_1,\tau_2;h)\sim F_0(\tau_1,\tau_2)+F_1(\tau_1,\tau_2)h+\cdots+F_j(\tau_1,\tau_2)h^{j}+\cdots,\qquad (\tau_1,\tau_2)\in\mathbb{R}^2,
$$
with $F_j(\tau_1,\tau_2)=\mathrm{const.}$ for $\tau_1+\tau_2\ge 1$, 
$F_0(\tau_1,\tau_2)=\sum_{j=1}^2 \alpha_j\tau_j+\mathcal{O}(|(\tau_1,\tau_2)|^2)$ as $|(\tau_1,\tau_2)|\to 0$
such that, for every fixed $\delta>0$, the eigenvalues of $P^{\mathrm{int}}$ in $(-\infty, E+h^{\delta})$
are of the form
$$
E+F\left(\Bigl(k_1+\frac12\Bigr)h,\Bigl(k_2+\frac12\Bigr)h;h\right)+\mathcal{O}(h^{\infty}),
\qquad k_1,k_2\in\mathbb{Z}_{\ge 0}.
$$
This result on the asymptotic distribution of the eigenvalues of $P^{\mathrm{int}}$ completes the proof of 
Theorem~\ref{thm-2}.
\end{proof}
Note that $F_1(0, 0)=0$.  
This is not explicitly stated in \cite{DS}, but follows from the proof given there when the subprincipal symbol vanishes 
at the nondegenerate minimum. Thus the eigenvalues of $P^{\mathrm{int}}$ and 
the real part of resonances of $P$ near $E$ are of the form 
\[
E+\alpha_1\Bigl(k_1+\frac12\Bigr)h+\alpha_2\Bigl(k_2+\frac12\Bigr)h+\mathcal{O}(h^2)
\]
under the assumption of Theorem~\ref{thm-2}.

\section*{Acknowledgement}
The first author is supported by JSPS KAKENHI Grant Number JP23KJ2090.

\section*{Statements and Declarations}
There are no competing interests to declare. 

\section*{Data availability}
We do not involve any data in this work.

\noindent
Department of Mathematical Sciences, Ritsumeikan University, 
1-1-1, Nojihigashi, Kusatsu-shi, Shiga, 525-8577, Japan

\noindent
E-mail address: kkameoka@fc.ritsumei.ac.jp

\vspace{0.3cm}
\noindent
Department of Robotics, Ritsumeikan University, 
1-1-1, Nojihigashi, Kusatsu-shi, Shiga, 525-8577, Japan

\noindent
E-mail address: n-yoshid@fc.ritsumei.ac.jp

\end{document}